\documentclass[11pt]{article}
\usepackage{amsfonts}
\usepackage{alltt}
\usepackage{amssymb,verbatim,ifthen}
\usepackage{dsfont}
\usepackage{pgfplots}
\usepackage[all]{xy}
\usepackage{color}
\usepackage{graphicx}
\usepackage{amsmath}
\usepackage{amsthm}
\usepackage{mathabx}	
\usepackage{pdfsync}
\usepackage{caption}
\usepackage{natbib}
\usepackage{thm-restate}
\usepackage{multirow}
\usepackage{bigints}
\usepackage{pgfplots}
\usepackage{tikz}
\usetikzlibrary{intersections,decorations.markings}
\usepgfplotslibrary{fillbetween}
\usepackage{filecontents}
\usepackage{subcaption}
\usetikzlibrary{spy}
\usetikzlibrary{math}

\usepackage{multirow,array}%multirow tables
\usepackage{pgfplots}%figures (feasble etc')
\pgfplotsset{compat=1.7}%the version of pgfplots to use

\pgfplotsset{compat=1.15}
\definecolor{Bluedudu}{RGB}{51,153,255}

\def\split{true}
\ifthenelse{\equal{\split}{true}}{}
{}
\newcommand{\ignore}[1]{}

\renewenvironment{proof}[1][Proof.]{\textbf{#1} }{\ \rule{0.5em}{0.5em}}

\newtheorem{theorem}{Theorem}
\newtheorem*{theorem*}{Theorem}

\newtheorem{lemma}{Lemma}

\newtheorem{definition}{Definition}

\newtheorem{proposition}{Proposition}

\newtheorem{remark}{Remark}

\def\1{{1\hskip-2.5pt{\rm l}}}

\def\eps{\varepsilon}

\def\Pr{\textbf{{\rm Pr}}}
\def\R{\mathbb{R}}
\def\N{\mathbb{N}}

\def\E{\mathbb{E}}

\def\p1{0.05}

\makeatother

\begin{document}

\title{Dynamic screening\thanks{For their valuable comments, the authors wish to thank the participants of the One World Mathematical Game Theory Seminar, the Aix-Marseille School of Economics Interaction seminar,  the Adam Smith Business School Micro-theory Seminar, the Manchester Economic Theory Seminar, ORSIS 2019, the Tel Aviv University Game Theory and Mathematical Economics Research Seminar, the Technion Game Theory Seminar, and the Bar Ilan University Game Theory  Seminar. \ignore{two anonymous referees, and the associate editor}
Lagziel acknowledges the support of the Israel Science Foundation, Grant \#513/19.}}
\author{David Lagziel\thanks{Ben-Gurion University of the Negev, Beer-Sheva 8410501, Israel.  E-mail: \textsf{Davidlag@bgu.ac.il}.} \
and  Ehud Lehrer\thanks{Tel Aviv University, Tel Aviv 6997801, Israel.  E-mail: \textsf{Lehrer@post.tau.ac.il}.}} 
\maketitle

\thispagestyle{empty}

\lineskip=2pt\baselineskip=5pt\lineskiplimit=0pt

\noindent{\textsc{Abstract}}:
\begin{quote}
We study dynamic screening problems in which elements are subjected to noisy evaluations and, at every stage,  some of the elements are rejected, whereas those remaining  are independently re-evaluated in subsequent stages.
We prove that, ceteris paribus, the quality of a screening process may not improve when the number of stages increases.
Specifically,  we examine the resulting elements' values and show that adding a single stage to a screening process may produce inferior results in terms of stochastic dominance, whereas increasing the number of stages substantially leads to a first-best outcome.
\end{quote}

\bigskip
\noindent {\emph{Journal of Economic Literature} classification numbers: C70, D49, D81.

\bigskip


\newpage

\lineskip=1.8pt\baselineskip=18pt\lineskiplimit=0pt \count0=1

%--------------------------------------------------------------------------
%--------------------------------------------------------------------------
\section{Introduction} \label{Section - Intro}
%--------------------------------------------------------------------------
%--------------------------------------------------------------------------

Imagine a bacterial infection that spreads through a population. 
Preventive treatment (for example, pre-exposure prophylaxis) is available,  but supply is limited.  %supply allows only a relatively small proportion of the population to receive it. 
Therefore,  the government seeks to identify and select high-risk individuals for this treatment. 
Fortunately,  a simple test, which is easy to administer across the population, provides a reasonable (though noisy) identification of individual risk. 
Accordingly, the existing policy grants treatment to individuals who score high on this test. 
Recently,  a more accurate---yet costly---test became available and was added as a second step to improve the screening. 
Specifically,  the first test is used to identify a subset of individuals as high-risk,  and they are then re-examined using the second test,  so that only the ones with the highest risk-score receive the treatment.
The final number of people receiving treatment remains as before.
After some experience with the new screening procedure,  it appears that the two-stage process did not improve the identification,  in fact,  it produced worse results altogether!
Namely, the combined test was more likely to misidentify low-risk individuals as high-risk ones compared to the one-stage procedure previously utilized.
How is this possible?

In the current paper we answer this question through the analysis of dynamic screening processes, such as the one described above. 
We focus on a decision maker who screens elements from a general set, based on noisy valuations.
The screening process could vary from a single to multiple stages,  while keeping an overall capacity constraint on accepted elements.
In other words,  once a dynamic screening process is in place,  a subset of the elements are rejected in every stage,  whereas the remaining elements are independently re-evaluated in subsequent stages.

Within this framework, we strive to identify the key advantages and disadvantages of dynamic screening.
We pursue this goal through two complementing paths.
The first is a comparison of a one-stage and a two-stage screening. 
The second path follows an asymptotic approach,  dealing with multi-stage screening in which the number of stages increases.
The joint results of these two approaches shed light on why do people use dynamic screening,  and under what conditions additional screening stages lead to suboptimal outcomes.

Specifically,  the first research path introduces a basic concept that we refer to as \emph{the hidden cost of dynamic screening}.
At face value,  additional screening stages impose a direct cost in time and effort.
Intuitively, one would expect that the added costs would be fully justified and compensated for by the anticipated superior nature of the obtained outcomes. 
We show, however, that even when the accuracy of the additional stages is superior, they may in fact generate \emph{inferior} results, both in terms of the expected value and in terms of stochastic dominance.
The basic idea behind this result is that given a capacity constraint on accepted elements,  the introduction of additional stages must be accompanied by some \emph{self-induced slack}.
One cannot simply add screening stages without either lowering the bar in preliminary stages and introducing suboptimal elements, or otherwise violating the capacity constraint.
Once these suboptimal elements are introduced, the noise in subsequent stages becomes more effective/destructive, although it is strictly more informative (in the sense of \cite{Lehmann1988}, for example) compared to previous stages.
Let us explain this insight through the following (stylized) example, which will also serve to present the key findings of our second research method.

\subsection{The hidden cost of dynamic screening: an example of self-induced slack} \label{Section - example in introduction}

Let us return to our preventive-treatment example.
To simplify the exposition, assume that the individual-risk parameter is uniformly distributed across the population, according to a random variable $V\sim U[0,1]$.
Because the preliminary test is somewhat noisy,  it cannot identify the realized value $v$ of every individual,  but only $v+N_1$, where $N_1\sim U[-1/4,1/4]$ is an additive unbiased noise variable (i.e., independent and symmetric about zero).
Assume further that the limited supply allows treatment for only $5\%$ of the population.
Following \cite{Lagziel2020a}, it is optimal for the government to use a threshold strategy.
Thus,  to meet the capacity constraint,  a threshold is fixed so that a realized noisy valuation grants treatment if and only if it is above that threshold.
This one-stage screening process yields a conditional distribution over the individual-risk parameter (across the population),  denoted $V_1$, whose CDF $F_{V_1}(v)$ is given by the solid (blue) graph in Figure \ref{Figure - example for one-stage vs two-stage screening}(a).

Now consider the updated two-stage policy which involves the  application of a superior second test.
Let us assume the second test is $20\%$ more accurate,  so it is characterized by an additive noise variable $N_2 \sim U[-1/5,1/5]$.
Again, due to the limited supply and high costs,  the government decides the administer the first test across the population,  and only the top $10\%$ of high-risk individuals are subjected to the second test.
Eventually, the $5\%$ treatment constraint is met.
This two-stage screening yields a conditional distribution of individual risk, denoted $V_2$, whose CDF $F_{V_2}(v)$ is given by the dashed (red) graph in Figure  \ref{Figure - example for one-stage vs two-stage screening}(a).

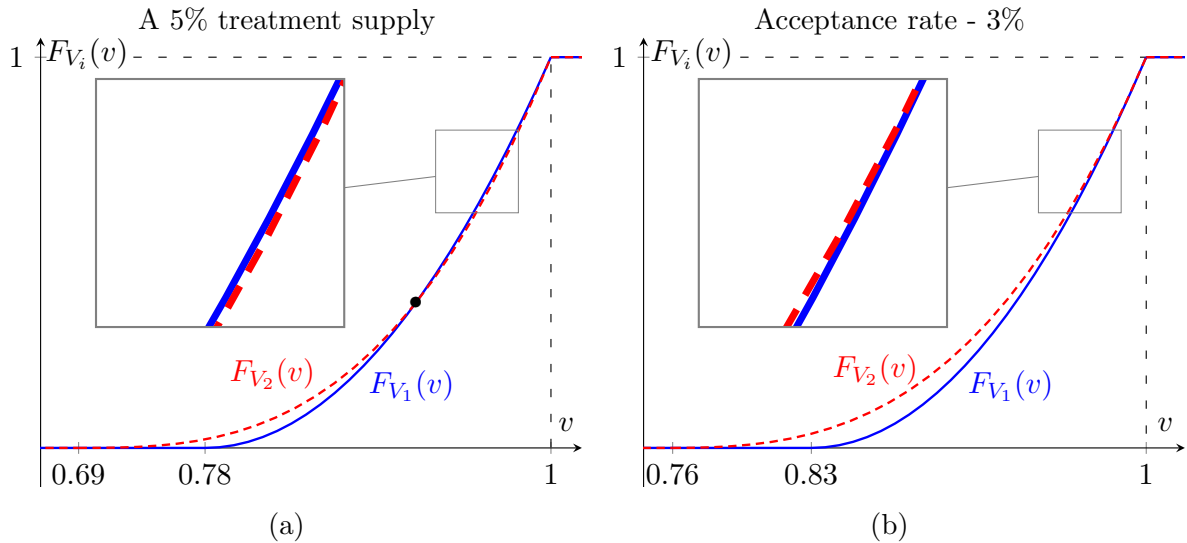
\begin{figure}[ht]
\centering
{\large One-stage versus two-stage screening}
\bigskip

\centering
\begin{minipage}{0.45\textwidth}

\centering 
A $5\%$ treatment supply

\begin{tikzpicture}[scale=1.05,spy using outlines={rectangle,gray,magnification=3,size=3.3cm, connect spies}]

\def\AAA{0.05}
\def\BBB{1}
\def\CC{0.3}

\begin{axis}[axis x line=middle, axis y line=middle, xtick={0,0.694595,1-(\AAA)^(1/2),1},ytick={0,1}, xmin=0.67, xmax=1.02, ymin=-0.1, ymax=1.05]

	\addplot[loosely dashed] coordinates{(0.73,1)(1,1)(1,0)};
  	\node at (axis cs: 1.01,0.06) {$v$};
   	\node at (axis cs: 0.7,1.01) {$F_{V_i}(v)$};

	\addplot[line width=\CC mm, solid, blue, domain=0:1-(\AAA)^(1/2)] {0};
	\addplot[line width=\CC mm,solid, blue, domain=1-(\AAA)^(1/2):1] {((1/\AAA)*(\x-(1-(\AAA)^(1/2)))^2)^(\BBB)};
	\addplot[line width=\CC mm, solid, blue, domain=1:1.05] {1};
	\node[blue] at (axis cs: 0.91,0.16) {$F_{V_1}(v)$};
	
	\addplot[line width=\CC mm,densely dashed, red, domain=0:0.694595] {0};
	\addplot[line width=\CC mm,densely dashed, red, domain=0.694595:1] {(50*(\x-0.694595)^2*(\x-(1-(0.1)^(1/2))-(1/3)*(\x-0.694595)))^(\BBB)};
	\addplot[line width=\CC mm,densely dashed, red, domain=1:1.05] {1};
	\node[red] at (axis cs: 0.82,0.19) {$F_{V_2}(v)$};

	\node [black,radius=1pt] at (0.912681,0.371487) {\small \textbullet};

	\spy on (5.8,4.2) in node [right] at (0.7,3.6);
		
\end{axis}
\end{tikzpicture}
\centering (a)

\end{minipage}
\hspace{1mm}
\begin{minipage}{0.45\textwidth}

\centering 
Acceptance rate - $3\%$

\begin{tikzpicture}[scale=1.05,spy using outlines={rectangle,gray,magnification=3,size=3.3cm, connect spies}]

\def\AAA{0.03}
\def\BBB{1}
\def\CC{0.3}

\begin{axis}[axis x line=middle, axis y line=middle, xtick={0,0.7550510,1-(\AAA)^(1/2),1},ytick={0,1}, xmin=0.74, xmax=1.02, ymin=-0.1, ymax=1.05]

 	\addplot[loosely dashed] coordinates{(0.79,1)(1,1)(1,0)};
  	\node at (axis cs: 1.01,0.06) {$v$};
   	\node at (axis cs: 0.765,1.01) {$F_{V_i}(v)$};

	\addplot[line width=\CC mm, solid, blue, domain=0:1-(\AAA)^(1/2)] {0};
	\addplot[line width=\CC mm,solid, blue, domain=1-(\AAA)^(1/2):1] {((1/\AAA)*(\x-(1-(\AAA)^(1/2)))^2)^(\BBB)};
	\addplot[line width=\CC mm, solid, blue, domain=1:1.05] {1};
	\node[blue] at (axis cs: 0.93,0.16) {$F_{V_1}(v)$};

	\addplot[line width=\CC mm,densely dashed, red, domain=0:0.7550510] {0};	
	\addplot[line width=\CC mm,densely dashed, red, domain=0.7550510:1] {((250/3)*(\x-0.7550510)^2*(\x-0.71835-(1/3)*(\x-0.7550510)))^(\BBB)};
	\addplot[line width=\CC mm,densely dashed, red, domain=1:1.05] {1};
	\node[red] at (axis cs: 0.86,0.20) {$F_{V_2}(v)$};

	\spy on (5.8,4.2) in node [right] at (0.7,3.6);

\end{axis}
\end{tikzpicture}
\centering (b)

\end{minipage}
\caption{\footnotesize  Both graphs relate to the example of one-stage and two-stage screening with $V\sim U[0,1]$, $N_1\sim U[-1/4,1/4]$, and $N_2 \sim U[-1/5,1/5]$.  The solid (blue) lines are the CDFs,  denoted $F_{V_1}(v)$,  among values under a one-stage screening,  while the dashed (red) lines are the CDFs, denoted $F_{V_2}(v)$, among  values under a two-stage process.  Though first-order stochastic dominance exists under a $3\%$ treatment constraint (or below), but not under a $5\%$ rate (see single crossing at solid black dot), the expected value under one-stage screening is strictly higher in both cases.}\label{Figure - example for one-stage vs two-stage screening}
\end{figure}

Computing the conditional expected value for both processes,  we get that $\E[V_1] > \E[V_2]$.
So not only the additional screening stage is costly, it actually generates an inferior outcome.
Moreover,  in a slightly worse situation, where the overall supply of treatment is just $3\%$ and the second test is only available for  $6\%$ of high-risk individuals, then the one-stage screening \emph{first-order stochastically dominates} the two-stage process [see the respective graphs in Figure  \ref{Figure - example for one-stage vs two-stage screening}(b)].

The cause for this phenomenon is the \emph{self-induced slack} of dynamic screening.
The introduction of additional screening stages must be accompanied by a reduced acceptance criterion (i.e., lower threshold) in preliminary stages. 
This change in criterion transfers additional burden into advanced screening stages in which both the noise and the sample space reduce \emph{simultaneously}, thus creating a considerable consequent effect on the outcomes.
%This transfers some additional burden into advanced stages, where the noise and the sample space reduce \emph{simultaneously}, which can largely affect the outcome.
In other words, there is a built-in trade-off in which the sample size and noise decrease at the same time,  so that the expected outcome could go either way.

Note that the concept of dynamic screening need not be dynamic in time.  
Specifically, one can get similar results when using several tests concurrently,  while following a unanimity rule.
Although technically not dynamic in time,  the unanimity  rule assures that the dual screening is essentially a dynamic process.
In recruitment processes, for example,  or even in  peer-review academic publishing, the use of two judges under a unanimity rule (for acceptance) is, de facto, a dynamic screening process.

What are the conditions that lead to sub-optimal  outcomes?
First and foremost, it is important to note that such results are not limited to uniform distributions (we provide similar results for general non-atomic distributions).
Second, as we already know,\footnote{See the notion of a \emph{contraction mapping} in \cite{Lagziel2020a}.} a one-stage screening process based solely on a more-accurate second test yields better results than a one-stage screening given the first, basic test.
Thus,  the fact that most of the screening is performed in less accurate stages plays an important role generating this phenomenon.
In fact, the effect is completely reversed when shifting from ``elite" screening (i.e., screening at the top of the distribution) to ``low-level" screening.
In other words, the implicit cost becomes evident when the overall capacity (i.e., acceptance) constraint is low, and subsequent stages are more vulnerable  to relatively high distributive noises.
If one should examine the previous examples given a high supply of treatment (for example,  more than $70\%$),  then the two-stage screening stochastically dominates the one-stage  process.

This distinction brings us to the second part of our analysis, i.e.,  the asymptotic approach.
Let us now consider a slightly different example: screening job applicants.
Because the stakes are high, many cutting-hedge institutions conduct prolonged applicants screening processes with multiple stages.
It seems quite implausible that all these institutions are under-performing.
To remove doubt, this is not a claim that we make.
Once multiple stages are introduced, then the overall constraint could be maintained by consistently screening in small portions.
That is,  every stage can support a high acceptance rate, effectively making it a low-level screening, whereas the final outcome matches the overall constraint.

In light of this insight, we reach our second main result which establishes a \emph{convergence to perfect screening}.
We prove that, even if all stages are subject to the same noise,  a multi-stage asymptotic screening process yields a first-best posterior distribution, as if the screening was preformed with no noise whatsoever.
We refer to such an outcome as \emph{perfect screening}, and show that standard stationary strategies,  such as a  fixed-threshold strategy or a fixed-capacity strategy,  lead to a perfect screening outcome.
Note that we obtain this result even when the same noise is used repeatedly,  that is, even if the accuracy along the stages does not improve.

A combination of the results that originate from these two research paths leads to a (somewhat striking) conclusion: the quality of a screening process, as a function of the number of stages, is not necessarily monotone.
To put it differently, introducing one additional referee or one additional stage may be detrimental,  whereas adding multiple ones would probably improve the screening.
In what follows,  we provide the specific details and conditions to support this statement.

\subsection{Relation to relevant literature}

The accurate position of this work in the literature is anything but trivial.
On the one hand,  this work joins the extensive theory of statistical decision making,  concerning both dynamic decision problems [going back to the works of \cite{Wald1939,Wald1947} and \cite{Arrow1949}] and the comparison of information structure.\footnote{See,  among others, \cite{Milgrom1981}, \cite{Quah_Strulovici2009}, \cite{Ganuza2010},  \cite{Chambers2011},  and more recently \cite{Athey2018}  and  \cite{Lagziel2020a}.}
On the other hand, this study also concerns the theory of information aggregation, which varies from social learning and information cascade to group decisions and committees of experts.\footnote{See, e.g.,  \cite{Banerjee1992}, \cite{Bikhchandani1992}, \cite{Glazer1998}, \cite{Dekel2000}, \cite{Ottaviani2001}, \cite{Levy2004},
\cite{Levy2007}, \cite{Visser2007}, \cite{Gershkov2009}, among many others. For a recent extensive survey, see \cite{Bikhchandani2021}.}

Though these two branches of the literature are extensive and thorough, our work precisely fits neither.
There are three basic elements that distinguish this research from previous ones.
First,  our formulation of screening problems, in the context of statistical decision theory,  naturally combines a capacity constraint that is typically missing from the aforementioned studies [with the exception of \cite{Lagziel2020a}].
This constraint is necessary for our analysis and outcomes.
Second,  many of the above-mentioned studies extended the seminal work of \cite{Condorcet1785} by introducing costly observations and strategic accumulation of information.
Our work is more basic in this sense, because it raises the question of whether another signal is beneficial altogether, irrespective of its price, and independently of the evaluators' preferences (thus information cascades and herding are less relevant in our framework).
Third, our framework builds on general, non-atomic distributions where the classification of a signal as either ``true"  or ``false" is irrelevant.
In our model, every signal provides more information about the actual value, yet as was exemplified  in Section \ref{Section - example in introduction}, the signal might still  be detrimental.
These stark differences are best exemplified by \cite{Bikhchandani2021}, who state that ``In purely individual decision making, an extra signal always makes an agent weakly better off" (see Section $2.6$ therein).
Though this statement is completely true for the relevant models of information aggregation and social learning,  one of our main results proves the opposite.

Nevertheless, there are some basic similarities between our study and previous ones.
The first, rather basic similarity,  is the fact that our asymptotic analysis yields the first-best screening outcome.
This is also the key insight of \cite{Condorcet1785}, with the obvious distinctions from our model:  \cite{Condorcet1785} and subsequent works build on a binary state of the world, a majority rule,  and a necessary condition concerning the informativeness of signals,  which are irrelevant in our framework.
Another similarity, mainly related to the research aspect,  arises from the studies of \cite{BenYashar2000} and \cite{Berend2005}, who prove that random committees of at least three experts outperform a single expert.
Yet, the key differences between the models and assumptions completely change the outcome under our analysis.

A different research field where our results may be of some interest is data analysis and data cleansing.
In the realm of statistical power analysis,  researchers need to maintain a sufficiently large sample size to detect an effect of a given size.
In this respect,  our capacity constraint fits rather naturally.  
Our results indicate that reducing the sample size does not necessarily improve the inspection of a given hypothesis.
Though a direct application goes well beyond the scope of this work (and quite possibility requires a different, extensive study altogether), we do believe that the given results may be of some importance to this field.

\subsection{Structure of the paper}
The paper proceeds as follows. 
In Section \ref{Section - Model} we describe the basic model and key definitions.
In Section \ref{Section - main results} we present the main results, divided into two subsections: 
in Subsection \ref{Section - one vs two stage screening} we compare one-stage and two-stage screening processes, and in Subsection \ref{Subsection - asymptotic screening problems} we carry out an asymptotic analysis of screening problems.
Concluding remarks, including a discussion about cost functions, are given in Section \ref{Section - extensions}.

%--------------------------------------------------------------------------
%--------------------------------------------------------------------------
\section{Preliminaries} \label{Section - Model}
%--------------------------------------------------------------------------
%--------------------------------------------------------------------------

%****** Outline in comments ******

%****** Opening paragraph - motivation ******

Consider a set of elements whose intrinsic values are distributed according to a non-constant random variable $V$, referred to as an  \emph{impact} variable.
For every $i \in \N$, let $N_i$ be a random variable which defines the additive evaluation errors in stage $i$, and referred to as the \emph{stage-i noise}.
We generally assume that all noise variables are unbiased --- symmetric about zero and jointly independent of $V$ and of each other.
A \emph{capacity} $p\in (0,1)$ dictates the proportion of accepted elements.
That is, the screening is constrained by the requirement to accept a fraction $p$ of the proposed elements.

A $k$-stage screening problem ${\rm SP} = \left(V, \{N_i\}_{i=1}^k, p\right)$ consists of an impact variable $V$, noise variables $N_1,\dots, N_k$, and a capacity $p$.
The screening problem evolves as follows.
Denote $V_1 = V$.
In each stage $i\geq 1$, the DM observes $V_i+N_i$ and fixes a threshold $t_i \in \R$, so that 
\begin{equation} \label{Eq. - Def. of cond. V}
V_{i+1} \sim V_i |\{V_i + N_i \geq t_i \},
\end{equation}
where
\begin{equation} \label{Eq. - Def. of capacity constraint}
\Pr(V_i + N_i \geq t_i) = p_i \ \ {\rm and } \ \ \prod_{i=1}^{k}p_i = p.
\end{equation}
In words, in every stage $i$, the DM observes the noisy valuation $V_i+N_i$ and fixes a screening threshold $t_i$, so that only the elements whose noisy valuations are at least $t_i$ proceed to the subsequent stage.
In every stage $i$, the DM maintains a capacity of  $p_i \in [0,1]$ to support an overall capacity of $p = \prod_{i=1}^{k}p_i$.
Note that we also allow for \emph{infinite} screening problems, denoted ${\rm SP} = \left(V, \{N_i\}_{i=1}^{\infty}, p\right)$, by taking an unbounded number of stages.

Given ${\rm SP} = \left(V, \{N_i\}_{i=1}^k, p\right)$, a \emph{strategy} $\tau$ is a sequence of threshold values $\tau = (t_1,\dots,t_k) \in \R^k$.
Let $V_{{\rm SP}}(\tau)$ denote the post-screening conditional distribution of the accepted elements' values.
That is, $V_{{\rm SP}}(\tau)= V_{k+1}$ where $V_{k+1}$ is defined according to Eq.\ \eqref{Eq. - Def. of cond. V}, Eq.\ \eqref{Eq. - Def. of capacity constraint},  and $\tau$.
The main goal of the DM is to maximize the expected value $ \E[V_{{\rm SP}}(\cdot)]$ of  accepted elements.

For tractability, we make the following assumptions: (i) every variable $X$ (of the above) has a density function $f_X$ and a CDF $F_X$ that are fully supported (namely,  $f_X$ is strictly positive) on some bounded interval, denoted $[\underline{X},\overline{X}]$; and (ii) unless stated otherwise, all noises are i.i.d.\ random variables.
As will later become evident, one can relax these assumptions (e.g., by varying the noises along the stages to be more informative in the sense of \cite{Lehmann1988}, or by using a contracting mapping as in \cite{Lagziel2020a}; see Subsection $3.2$ therein)  and still maintain our key insights.
To be clear,  we do divert from the second assumption,  regarding i.i.d.\ noises,  in the first part of our analysis when we examine a two-stage screening with a more informative second stage,  relative to the first stage.

\subsection{Stationary strategies} \label{Section - stationary strategies}

Our analysis is based on threshold strategies, and in the context of dynamic screening, one can also consider two types of stationary strategies: fixed-threshold strategies and fixed-capacity strategies.
Formally, fix a $k$-stage screening problem ${\rm SP}$, and consider the following two  stationary strategies.
The first strategy, referred to as the \emph{fixed-threshold strategy}, dictates that $\tau = (t,\dots,t)$, which means that all threshold values, throughout the $k$ stages, are identical.
By continuity, one can fix the (unique) threshold value $t$ to maintain the capacity constraint $p$, thus the strategy is well defined.
The second strategy, referred to as the \emph{fixed-capacity strategy}, dictates that $p_i = p^{1/k}$ for every stage $1\leq i \leq k$. 
That is,  all threshold values are fixed so that, in every stage, a fraction $p^{1/k}$ of the elements being evaluated proceed to the subsequent stage.
We shall use these strategies in both parts of our analysis.

The first part of our analysis also requires a precise definition concerning stochastic dominance of one screening process over another.
The following definition captures this notion.

\begin{definition}
For every $i=1,2$, consider an $i$-stage screening problem ${\rm SP}_i$ with a specific strategy $\tau_i$.
We say that \emph{one-stage screening under ${\rm SP}_1$ given $\tau_1$ stochastically dominates two-stage screening under ${\rm SP}_2$ given $\tau_2$}  if $V_{{\rm SP}_1}(\tau_1)$ first-order stochastically dominates $V_{{\rm SP}_2}(\tau_2)$.
\end{definition}

We should clarify that a one-stage screening problem requires no specification regarding the strategy, since the relevant threshold is well defined and unique.
Therefore,  henceforth we will not specify the type of the one-stage strategies.

\section{Main results} \label{Section - main results}

Our analysis consists of two parts:
in Subsection \ref{Section - one vs two stage screening} we compare one-stage screening with a two-stage process, and in Subsection \ref{Subsection - asymptotic screening problems} we adopt an asymptotic approach by focusing on a dynamic screening while substantially increasing the number of stages.

\subsection{A comparison of one-stage and  two-stage screening} \label{Section - one vs two stage screening}

Our comparison of a one- and a two-stage screening processes consists of four results.
First, in Theorem \ref{Theorem - FOSD with fixed thresholds}, we consider a two-stage screening performed under a stationary, fixed-threshold strategy where all noises are identically distributed.
Under these conditions we prove that for a sufficiently low acceptance rate (i.e., elite screening) one-stage screening stochastically dominates the two-stage process.
Next, we extend this result in Proposition \ref{Proposition - FOSD distinct screening} and in Proposition \ref{Proposition - FOSD full distinct screening} by omitting the stationary fixed-threshold condition and by allowing the noises' distributions to vary between the two stages.
Lastly,  in Theorem \ref{Theorem - two-stage dominance with fixed thresholds}, we show how one can revert these results by adopting a sufficiently high capacity constraint, i.e.,  through low-level screening.

We start with Theorem \ref{Theorem - FOSD with fixed thresholds},  which relates to i.i.d.\ noises along with an implementation of a stationary, fixed-threshold strategy in the two-stage process.
In other words,  Theorem \ref{Theorem - FOSD with fixed thresholds} captures the alternative interpretation (for dynamic screening) of using two evaluators instead of one, when both are subjected to the same noise and use the same screening strategy.
The theorem states that, for a sufficiently low capacity, namely a very restrictive screening,  it is better to use one stage/evaluator instead of two.  
(All proofs are deferred to the Appendix.)

\begin{theorem} \label{Theorem - FOSD with fixed thresholds}
For every impact and noise variables $V$ and $N$,   there exists $p_0>0$ such that for every capacity $p < p_0$,  one-stage screening under $(V,N,p)$ stochastically dominates two-stage screening under $(V,\{ N,N \},p)$ given a stationary fixed-threshold strategy.
\end{theorem}

Though the proof is given in the Appendix, we wish to provide here some intuition on its structure and technique to better understand the result.
In the proof we consider the probability density functions $(f_{V_1},f_{V_2})$ of the one- and the two-stage screening processes, respectively.
In order to preserve the capacity constraint,  the two-stage process must follow a lower threshold value,  so that $f_{V_2}$ is supported on a larger interval and some lower values are generated with positive probabilities, whereas similar values are eliminated under the one-stage screening.
So, to establish (first-order) stochastic dominance,  we prove that the graphs of the two densities intersect only once.
Graph $(a)$ in Figure \ref{Figure - one vs two stage screening} provides some intuition for this. 

Note that that two-stage screening in Figure \ref{Figure - one vs two stage screening} translates to a parabolic graph,  rather than a straight line, by the fact that the probability of passing two independents tests is the product of two probabilities, one for each stage.
Moreover, notice how the capacity plays a key role in this analysis, as stricter screening (namely, a smaller capacity $p$) increases the threshold values and shifts the two curves to the right, thus maintaining the single crossing property and the dominance of one-stage screening over the two-stage process.
On the other hand,  a more lenient screening (i.e., a higher capacity $p$) reduces the threshold criteria and shifts both curves to the left. 
This allows the RHS crossing (of the two curves) to emerge, while eliminating the LHS crossing from the support of $V$, so that two-stage screening dominates the one-stage process.
(This is the main result of Theorem \ref{Theorem - two-stage dominance with fixed thresholds} below.)

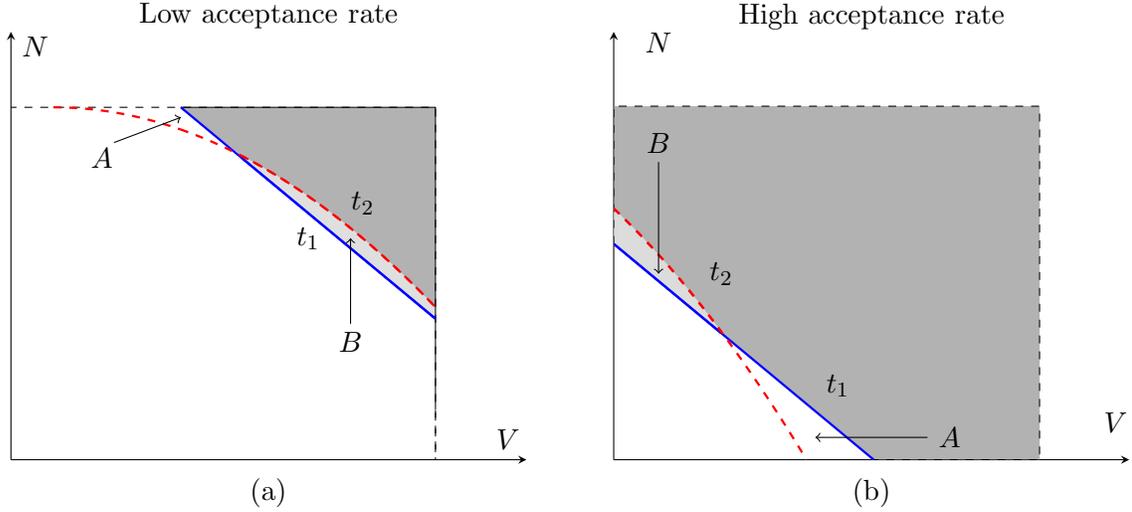
\begin{figure}[ht]
\centering
{\large One-stage versus two-stage screening on the $(V,N)$-plain}
\bigskip

\def\CC{0.3}

\centering
\begin{minipage}{0.45\textwidth}

\centering 
Low acceptance rate

\begin{tikzpicture}%[scale=1.1]
\begin{axis}[axis x line=middle, axis y line=middle, xtick={0},ytick={0}, xmin=0.3, xmax=1.15, ymin=0.3, ymax=1.15]

%Add lines
	\addplot[name path=G, line width=0pt] coordinates{(0.58,1)(1,1)(1,0.4)};	
	\addplot[name path=A, dashed] coordinates{(1,0.1)(1,1)(0.2,1)};	
	\addplot[dashed ] coordinates{(0,1)(0.2,1)};	
	\addplot[dashed] coordinates{(1,0)(1,0.4)};	
	\addplot[name path=B, line width=\CC mm, solid, blue, domain=0.58:1] {1.58-\x};
	\addplot[name path=E, line width=\CC mm, solid, blue, domain=0.68:1] {1.58-\x};
	\addplot[name path=D, line width=\CC mm, dashed, red,domain=0.37:0.58] {1-(\x-0.37)^2};
	\addplot[name path=F, line width=\CC mm, dashed, red, domain=0.68:1] {1-(\x-0.37)^2};
	\addplot[name path=C, line width=\CC mm, dashed, red,domain=0.58:1] {1-(\x-0.37)^2};

%Add color fill
	\addplot[fill=gray, opacity=0.6] fill between[of=G and F];
	\addplot[fill=gray, opacity=0.25] fill between[of=F and E];
		
%Add text
	\node at (axis cs: 1.12,0.34) {$V$};
	\node at (axis cs: 0.34,1.12) {$N$};
	\node at (axis cs: 0.88,0.81) {$t_2$};
	\node at (axis cs: 0.79,0.74) {$t_1$};
	\node at (axis cs: 0.45,0.9) {$A$};
	\node at (axis cs: 0.86,0.535) {$B$};

%Add arrows
		\draw [->] (0.86,0.57) -- (0.86,0.74);
		\draw [->] (0.47,0.93) -- (0.58,0.98);

\end{axis}

\end{tikzpicture}

\centering (a)

\end{minipage}
\hspace{1mm}
\begin{minipage}{0.45\textwidth}
\centering 
High acceptance rate
\begin{tikzpicture}%[scale=1.1]
\begin{axis}[axis x line=middle, axis y line=middle, xtick={0},ytick={0}, xmin=0, xmax=1.15, ymin=0, ymax=1.15]

%Add lines
	\addplot[name path=A, dashed] coordinates{(0,0.58)(0,0.95)(0.95,0.95)(0.95,0)(0.58,0)};	
	\addplot[name path=B, line width=\CC mm, solid, blue, domain=0:0.58] {0.58-\x};
 	\addplot[name path=E, line width=\CC mm, solid, blue, domain=0:0.24] {0.58-\x};
	\addplot[name path=D, line width=\CC mm, dashed, red,domain=0:0.43] {1-(\x+0.57)^2};
	\addplot[name path=F, line width=\CC mm, dashed, red, domain=0:0.24] {1-(\x+0.57)^2};
%	\addplot[name path=C, line width=\CC mm, dashed, red,domain=0.58:1] {1-(\x+0.42)^2};

%Add color fill
 	\addplot[fill=gray, opacity=0.6] fill between[of=B and A];
	\addplot[fill=white, opacity=0.55] fill between[of=E and F];
		
%Add text
	\node at (axis cs: 1.12,0.1) {$V$};
	\node at (axis cs: 0.1,1.12) {$N$};
	\node at (axis cs: 0.24,0.5) {$t_2$};
	\node at (axis cs: 0.5,0.2) {$t_1$};
	\node at (axis cs: 0.1,0.85) {$B$};
	\node at (axis cs: 0.75,0.06) {$A$};

%Add arrows
		\draw [->] (0.1,0.8) -- (0.1,0.5);
		\draw [->] (0.7,0.06) -- (0.45,0.06);

\end{axis}

\end{tikzpicture}

\centering (b)

\end{minipage}
\caption{\footnotesize Both graphs describe a comparison between one- and two-stage screening: Graph $(a)$ provides such a comparison under a low acceptance rate (a small $p$), whereas Graph $(b)$ provides a similar comparison under a high acceptance rate (i.e., a high $p$).  Under the one-stage screening, depicted by the straight and solid (blue) lines $t_1$, only grey areas pass the screening.  Under the two-stage screening, depicted by the parabolic and dashed (red) lines $t_2$, the light grey area $B$ is eliminated, while all elements in the white area $A$ pass the screening. } \label{Figure - one vs two stage screening}

\end{figure}

Going beyond the technicalities, Theorem \ref{Theorem - FOSD with fixed thresholds} is based on the idea that an additional stage, under an overall capacity constraint, must be accompanied with lower thresholds, relative to the one-stage screening.
In the second stage this reduction poses a problem since the smaller sample size augments the noise's effectiveness in distorting the underlying distribution of values.
That is, once lower values pass through the first stage, out of the two-stage process,  they have a higher probability to pass through the second stage as well, since the noise is relatively more effective, given the smaller sample size.
This does not happen in the one-stage process, since these lower values are eliminated completely due to the higher threshold.

\begin{remark}
One can extend the result of \emph{Theorem \ref{Theorem - FOSD with fixed thresholds}} beyond the assumptions of bounded supports and positive densities.
For example,  \emph{Theorem \ref{Theorem - FOSD with fixed thresholds}} also holds for noises whose probability density functions monotonically diminish from a certain point, or even for noises with unbounded supports whose probability density functions decrease sufficiently fast.
By its technical nature,  the full characterization is left for future research.
\end{remark}

The next two propositions extend Theorem \ref{Theorem - FOSD with fixed thresholds} in two ways.
Proposition \ref{Proposition - FOSD distinct screening} accounts for non-stationary strategies while maintaining i.i.d.\ noises, and Proposition \ref{Proposition - FOSD full distinct screening} extends Proposition \ref{Proposition - FOSD distinct screening} by also accounting for different noises.
These two extensions require some limitations concerning the two-stage screening strategy.
Specifically,  by allowing \emph{any} two-stage strategy, one can effectively converge to a one-stage process if one threshold is sufficiently low, which would make it redundant.
Moreover, if the noise in the second stage is a contraction of the noise in the first stage, then we know that doing most of the screening in the second stage will guarantee a better expected outcome compared to a one-stage process.
Therefore, we introduce the following definitions of \emph{distinct} two-stage strategies.

\begin{definition} \label{Definition - distinct strategies}
Fix $\eps \in (0,1/2)$ and a two-stage screening problem $\rm{SP}_2$ given a strategy $\tau=(t_1,t_2)$. 
Note that $\tau$ defines two specific capacities $(p_1,p_2)$.
We say that $\tau$ is \emph{$\eps$-distinct} \emph{(}from a one-stage strategy\emph{)} if $\max{\{p_1,p_2\}}<1-\eps.$
In addition,  $\tau$ is \emph{fully $\eps$-distinct}  if $\eps < p_2<1-\eps$.
\end{definition}

In simple terms, a two-stage screening strategy is distinct from a one-stage strategy if no stage is redundant in the sense that it screens less than a proportion $\eps$ of the elements.
Moreover, a fully distinct strategy sustains a constraint on the second stage such that, on the one hand, the second stage is not redundant and, on the other hand, not all the screening occurs in the second stage.
As stated, these constraints are necessary, otherwise one could effectively eliminate one stage from the two-stage process and generate results that are at least as good as in one-stage screening.

Though restrictive, there is a simple reasoning as to why the mentioned limitations are rather natural and evident.
Generally speaking, the advanced stages in many screening processes are more accurate,  therefore more costly.
For example,  interviews with chief executives or seminars in academic institutions consume a lot of time and effort from busy, time-constrained individuals.
It would be very costly to conduct most of the screening in these stages rather than eliminate most of the applicants in preliminary stages., e.g., when screening CVs.
So, one must limit the capacity of elements that reach these advanced stages, and Definition \ref{Definition - distinct strategies} provides such a limitation.
We discuss this aspect more broadly in Section \ref{Section - extensions}.

The following Proposition \ref{Proposition - FOSD distinct screening} shows that even if one should divert from stationary fixed-threshold strategies,  the dominance of one-stage screening over the two-stage process remains valid, provided that the strategy is $\eps$-distinct and the capacity is sufficiently small.

\begin{proposition} \label{Proposition - FOSD distinct screening}
Fix $\eps \in (0,1/2)$.
For every impact and noise variables $V$ and $N$,   there exists $p_{\eps}>0$ such that for every capacity $p < p_{\eps}$, one-stage screening under $(V,N,p)$ stochastically dominates two-stage screening under $(V,\{ N,N \},p)$ given an $\eps$-distinct strategy.
\end{proposition}

The next proposition shows that one can also extend Proposition \ref{Proposition - FOSD distinct screening} to different noises in the two-stage process.
For this to hold, we consider a fully distinct strategy such that the second stage does not account for most of the process.

\begin{proposition} \label{Proposition - FOSD full distinct screening}
Fix $\eps \in (0,1/2)$,  an impact variable $V$,  and two noise variables $N_1$ and $N_2$. 
There exists $p_{\eps}>0$ such that for every capacity $p < p_{\eps}$, one-stage screening under $(V,N_1,p)$ stochastically dominates two-stage screening under $(V,\{ N_1,N_2 \},p)$ given any fully $\eps$-distinct strategy.
\end{proposition}

Let us emphasize that,  in Proposition \ref{Proposition - FOSD full distinct screening},  we do not limit the distribution of $N_2$ relative to $N_1$,  other than the general condition of a strictly positive density on some interval.
This is more than a mere technicality, since it shows that the additional stage can still distort the screening, even if the added stage is very accurate (i.e., even if the noise is relatively mild).
In other words, one does not need to devise esoteric distributions to exemplify our results, but simply extend the screening to the top of the distribution.

On the other hand, if the decision maker wishes to extract  only  a small portion from the bottom of the distribution, then two-stage screening becomes superior relative to the one-stage process.
The intuition is that  once the sample size in both stages is significantly large, the noises' influence becomes limited, which makes the additional stage worthwhile.
The following theorem is a mirror image of Theorem \ref{Theorem - FOSD with fixed thresholds} for cases in which the decision maker focuses on low-level screening (i.e., whenever the capacity is relatively high), showing that two-stage screening dominates the one-stage process.

\begin{theorem} \label{Theorem - two-stage dominance with fixed thresholds}
For every impact and noise variables $V$ and $N$,  there exists $p_0>0$ such that for every capacity $p > p_0$,  two-stage screening under $(V,\{ N,N \},p)$ given a stationary fixed-threshold strategy stochastically dominates one-stage screening under $(V,N,p)$.
\end{theorem}

The combination of Theorem \ref{Theorem - FOSD with fixed thresholds} and Theorem  \ref{Theorem - two-stage dominance with fixed thresholds} implies that the 
superiority of one screening method over another, varying in the number of stages, greatly depends on the capacity constraint.
The transition from a sufficiently low  to a sufficiently high capacity, given a stationary fixed-threshold strategy,  exemplifies how an additional screening stage changes from a burden to an advantage. 
A question that remains for future research is whether this transition occurs at a single point such that one-stage screening is superior below a given capacity and inferior above it, or whether this transition occurs in multiple points.

We now proceed to the second part of our analysis to prove that adding a considerable number of stages strictly improves a screening process.

\subsection{Convergence to perfect screening} \label{Subsection - asymptotic screening problems}

The results in Section \ref{Section - one vs two stage screening} may provide the false impression that dynamic screening is inefficient.
In this section we prove that this conclusion is false by showing that a multi-stage process eventually yields the first best outcome.
Theorem \ref{Theorem - two-stage dominance with fixed thresholds} provides some intuition for this,  since a sufficiently high capacity ensures that additional stages only improve the screening,  and that is indeed the case when using multiple stages.
We establish this conclusion through two supporting results.
The first,  Proposition \ref{Proposition - increasing strategy}, concerns infinite screening problems and shows that any increasing strategy (i.e., a strategy under which threshold values can only increase along the stages) generates the first-best outcome.
The second,  Theorem \ref{Theorem - stationary strategies converge to perfect screening}, shows that the two previously mentioned stationary strategies generate posterior distributions that converge, in distribution, to the first best result.

\subsubsection{A Perfect screening strategy} \label{Subsection - infinite screening problems}

In every screening problem, the best the DM can strive for is a screening procedure that yields a result as if there is no noise whatsoever --- a result that we refer to as  \emph{a perfect screening}.
Formally, given a screening problem ${\rm SP}$, a strategy $\tau$ yields a \emph{perfect screening} if $\E[V_{{\rm SP}}(\tau)]  = \E[V| V\geq v_p]$, where $v_p$ denotes the $p$-quantile of $V$ (i.e., $\Pr(V\geq v_p) =p$).
In other words, a \emph{perfect screening} strategy induces an expected value which is equivalent to a screening without noise, while maintaining the same capacity constraint.
Note that a perfect screening of $\tau$, under a capacity of $p$, entails that the two variables  $V_{{\rm SP}}(\tau)$ and $V| \{V\geq v_p\}$ are identically distributed.

Starting with infinite screening problems, we say that an infinite strategy $\tau = (t_1,t_2,\dots) \in \R^{\infty}$ is \emph{increasing} if $t_{k+1} \geq t_k$ for every $k\geq 1$.
That is, an increasing strategy entails that the screening along the stages becomes stricter.
The following proposition shows that, in every infinite screening problem, every increasing strategy produces a perfect screening.

\begin{proposition} \label{Proposition - increasing strategy}
In an infinite screening problem, every increasing strategy induces a perfect screening.
\end{proposition}

The motivation behind the statement and proof of Proposition \ref{Proposition - increasing strategy} originates from the fact that suboptimal elements are discarded with some ``patience" on the side of the DM.
So even if the noise is rather disruptive for screening, for example, an almost-binary noise with a large variance,\footnote{This type of noises (and others) are prone to screening biases; see \cite{Lagziel2019} for more details.} a perfect screening remains feasible since suboptimal elements are slowly screened throughout the stages.

\subsubsection{Stationary strategies converge to perfect screening} \label{Subsection - Staionary strategies converge}

After establishing that every infinite and increasing strategy yields a perfect screening, one may want to consider a more practical finite set-up.
In practice,  whether we consider screening job applicants or ``cleansing" datasets, decision makers cannot feasibly commit to infinite screening stages,  which makes the  finite set-up the only practical choice.
In such scenarios, a basic question is whether simple finite strategies converge to a perfect screening.
In this section we tackle this question and provide two stationary strategies that converge to a perfect screening.

As already defined in Subsection \ref{Section - stationary strategies}, we consider the \emph{fixed-threshold strategy} which maintains the same threshold value throughout the stages, and we consider the \emph{fixed-capacity strategy} which dictates the same capacity in all stages.
For each of these strategies we shall prove that the induced distribution converges, in distribution, to a perfect screening outcome.
More formally,  we say that a stationary strategy $\tau$ \emph{converges to a perfect screening} if, for every $k$-stage screening problem ${\rm SP}=(V,\{N\}_{i=1}^k,p)$, we have that $V_{{\rm SP}}(\tau) \xrightarrow[\ \ \ \ ]{d} V|\{V\geq v_p\}$ as $k \to \infty$.
In light of Proposition \ref{Proposition - increasing strategy},  note that the fixed-capacity strategy does not converge to an infinite increasing strategy.
In fact, the limit of the fixed-capacity strategy is not well defined for an infinite screening problem with a capacity constraint of $p\in (0,1)$.
Therefore,  establishing its convergence to a perfect screening outcome requires a separate approach and proof, given in the following Theorem \ref{Theorem - stationary strategies converge to perfect screening}.

\begin{theorem} \label{Theorem - stationary strategies converge to perfect screening}
The fixed-threshold strategy and the fixed-capacity strategy converge to a perfect screening.
\end{theorem}

The use of stationary strategies accommodates a gradual screening process  such that only considerably low valuations are eliminated, while most elements proceed to subsequent stages.
Though one can devise different non-stationary strategies that converge to a perfect screening,  it is clear that not all strategies will do so.
For example,  consider a screening strategy under which the first-stage threshold is too high, for example,  $t_1 = v_p+\delta$ for some small $\delta>0$, such that the capacity constraint is not violated, yet some values above $v_p$ are partially eliminated in the first stage.
In such a case,  the thresholds in subsequent stages must be lower, and the resulting posterior distribution will not converge to $V|\{V\geq v_p\}$.

\section{Concluding remarks} \label{Section - extensions}

\subsection{The cost of screening}

Though our comparison of screening processes does not explicitly incorporate a cost function, we do not remain naive for this consideration.
The basic intuition is that additional stages and improved accuracy are more costly, so one should balance the two criteria with the superior expected results.
For example,  consider a DM who fixes a screening process given that advanced stages are more accurate (under some metric) than previous ones.
A cost-minimization analysis would typically dictate that most of the screening,  in terms of capacities, is performed in preliminary stages rather than in advanced ones.
In other words, since the screening becomes more costly,  the DM limits the mass of elements that reaches advanced stages.
Therefore,  one can focus on the capacities as an implicit proxy/measure for the needed stages and accuracy.
Interestingly, our results in Subsection \ref{Section - one vs two stage screening}, namely Theorem \ref{Theorem - FOSD with fixed thresholds} and Propositions \ref{Proposition - FOSD distinct screening} and \ref{Proposition - FOSD full distinct screening}, implicitly contain this feature even without an explicit cost function.

More formally, consider a $k$-stage  screening process with the following, illustrative cost function  $\sum_{i} \alpha_i \log\left( \tfrac{\prod_{j\leq i-1} p_j}{p} \right)$, where $p_0=1$ and $\alpha_{i+1} \geq \alpha_i >0$ are accuracy indices for stages $i+1$ and $i$, respectively (i.e., a higher $\alpha$ indicates a more informative screening under some metric, as in \cite{Lehmann1988} or \cite{Lagziel2020a}, among others).
The intuition behind this function is that  the cost of every stage $i$ increases with respect to: (i) the informativeness $\alpha_i$; and (ii) the mass of the inspected elements $ \prod_{j=0}^{i-1} p_j$.
For simplicity,  consider $k=2$ and assume the DM is bounded by a binding budget constraint $C>0$. 
So, to meet  the condition $\alpha_1 \log\left(\tfrac{1}{p}\right) + \alpha_2 \log\left(\tfrac{p_1}{p}\right) \leq C$,  the DM must ensure $p_1$ is sufficiently close to $p$, which in turn pushes $p_2$ towards $1$.
It is straightforward to verify that this constraint becomes even more binding as $\alpha_2$ increases.
%Similarly, one could extend this reasoning to more than two stages, with increasing levels of accuracy.

One could also take a different perspective altogether, and consider a set-up where the accuracy of every stage/test, and therefore the cost,  depend directly on the induced capacities.
For example, consider the extreme case where an examination is performed in stage $i$ under a capacity constraint of $p_i=1$.
What would be the cost in this case?
Evidently, the cost should be zero because, effectively, there is no screening.
Thus, one can consider another illustrative cost function of the form $-\sum_i \log(p_i)$, so that the cost of not preforming a screening is zero (that is, $\log(p_i)=0$ if and only if $p_i=1$), and it increases as the capacity decreases.
In this example, it is clear that $-\sum_i \log(p_i)=-\log(p)$ and the cost depends only on the overall capacity constraint.
In such cases,  the DM would be solely concerned with comparing the outcomes of the one-stage and the $k$-stage screening processes.

\subsection{Summary}

In this paper we presented an analysis of dynamic screening, showing that the quality of the screening process is not necessarily a monotone function in the number of stages.
Specifically, one can add a single screening stage such that the overall quality of the screening process decreases, whereas a few additional stages significantly improve the process.
There are a few natural questions that arise from our analysis.
First, what exactly occurs when reverting from a single-stage process to a high number of stages? 
It seems reasonable that first-order stochastic dominance does not revert in a single stage, but that there is a slow transition from one posterior distribution to another.
Identifying the posterior distribution's evolution as a function of the number of stages is an important follow-up question, especially for applicative purposes.

Another rather difficult question to tackle in future research relates to the nature of the noise throughout the process.
In our framework,  noises are given exogenously, while in practice noises are endogenously determined according to the the capacities and feasibility constraints (technical and monetary).
Since the technical complexity of such questions is potentially overwhelming, it might be essential (and even more interesting) to adopt a combined empirical-theoretical perspective, where the assumptions and theoretical analysis are based on actual data.

%\vspace{4cm}

% ---------------------------------------

%\bibliographystyle{aer}
\bibliographystyle{ecta}
%\bibliographystyle{plain}
%\bibliography{General_citations}
\bibliography{./MyCollection}
%\bibliography{../../General_citations}

% ---------------------------------------

\appendix

\section{Appendices}

\subsection{Proof of Theorem \ref{Theorem - FOSD with fixed thresholds}}
\begin{proof}
Fix $V,N$,  and  some capacity $p \in (0,1)$.
Denote the one-stage and the two-stage screening problems by ${\rm{SP}}_i=(V,\{N\}_{k=1}^i,p)$,  where $i=1,2$ respectively.
Denote the screening strategy (and threshold) of $\rm{SP}_1$ and $\rm{SP}_2$ by $\tau_1=(t_1)$ and  $\tau_2=(t_2,t_2)$, respectively.
Recall that we consider a two-stage screening process with a fixed-threshold stationary strategy.
Note that
\begin{eqnarray*}
F_{V_1}(t)
& = & \Pr \left(V \leq t | V + N \geq t_{1}\right) \\
& = & \tfrac{1}{p} \Pr \left(V\leq t , V+N \geq t_1 \right) \\
& = & \tfrac{1}{p} \int_{-\infty}^t f_{V}(x) \Pr\left(N \geq t_{1} - x\right) dx. 
\end{eqnarray*}
Thus,
\begin{eqnarray*}
f_{V_1}(t)
& = & \tfrac{1}{p} f_{V}(t) \Pr\left(N \geq t_{1} - t \right) \\
%& = & \tfrac{1}{p} f_{V}(t) G\left( t_{1} - t \right) \\
& = & 
\begin{cases}
	\tfrac{1}{p} f_{V}(t) G\left( t_{1} - t \right),  & \text{ for } t_1 - \overline{N} \leq t \leq  \overline{V}, \\
	0, 							         										   & \text{ otherwise,} 
\end{cases}
\end{eqnarray*}
where $G(x) = \Pr(N \geq x)$ is a differentiable decreasing function,  such that $G(x)=0$ if $x \geq \overline{N}$, and $G(x)=1$ if $x \leq \underline{N}$.
Similarly,  
\begin{eqnarray*}
f_{V_2}(t)
& = & \tfrac{1}{p} f_{V}(t) G^2\left( t_{2} - t \right) \\
& = & 
\begin{cases}
	\tfrac{1}{p} f_{V}(t) G^2\left( t_{2} - t \right),  & \text{ for } t_2 - \overline{N} \leq t \leq  \overline{V}, \\
	0, 							         										   & \text{ otherwise.} 
\end{cases}
\end{eqnarray*}
The capacity constraint implies that $t_2 < t_1$, and both values converge to $\overline{V}+\overline{N}$ as $p \to 0$. 
So, we will prove that $V_1$  first-order stochastically dominates (FOSD) $V_2$ by showing that the exists a range of values close to $\overline{V} + \overline{N}$ (equivalently,  a range for $p$ close to zero), such that for every $t_1$ and every $t_2<t_1$ in that interval,  there exists a unique (interior) point $t \in (t_1 - \overline{N},  \overline{V})$ such that $G(t_1-t) = G^2(t_2-t)$, namely a single crossing between the two densities.

To simplify the analysis,  define the function $J(x)$ to be the linear extension of $G$ outside the interval $[\underline{N},\overline{N}]$.
Specifically,
$$
J(x) =
\begin{cases}
	G(x),  & \text{ for } x\in [\underline{N},\overline{N}], \\
	d_{\overline{N}} \left(x-\overline{N}\right) ,  & \text{ for } x\geq \overline{N}, \\
	d_{\underline{N}} \left(x-\underline{N}\right) +1  & \text{ for } x \leq \underline{N},
\end{cases}
$$
where $d_{\overline{N}}= \lim_{x\to \overline{N}^-}\frac{G(x)-G(\overline{N})}{x-\overline{N}}$ and $d_{\underline{N}}= \lim_{x\to \underline{N}^+}\frac{G(x)-G(\underline{N})}{x-\underline{N}}$.
Thus, $J'(\overline{N})<0$ is well defined and strictly negative, since $G$ is strictly decreasing on $[\underline{N},\overline{N}]$,  by the assumption that all variables have strictly positive densities on their compact supports.

Define the function
$$
H(t,t_2)= J(t_1-t) -J^2(t_2-t).
$$
With a slight abuse of notation, we regard $t_2$ as a variable, independent of either $t_1$, or $p$.
%For $t\in (t_1 - \overline{N},t_1 - \underline{N})$, we get $H(t,t_1)= J(t_1-t) -J^2(t_1-t) >0$.
Note that $H$ is continuously differentiable,  $H(t_1-\overline{N},t_1)=0$, and 
$H'_t(t_1-\overline{N},t_1) = -J'(\overline{N}) +2J(\overline{N})J'(\overline{N}) = -J'(\overline{N})>0$.
Thus, we can use the implicit function theorem to establish that there exist a continuously differentiable function $T(\cdot)$ defined on an open interval $I_0$ around $t_1$, such that $H(T(t_2),t_2)=0$ for every $t_2 \in I_0$, and $T(t_1)=t_1-\overline{N}$.

Using this result, let us now prove that there exists a neighbourhood $I_1 = (t_1 - \eps, t_1 +\eps) \subseteq I_0$, where $\eps>0$, such that for every $t_2 \in I_1$, there exists a unique solution $t\in T(I_1)$ for $H(t,t_2)=0$.
Clearly, a solution exists according to the previous result, so we can assume, to the contrary, that for every interval $I_1 \subseteq I_0$ there exists $t_2\in I_1$ so that the equation $H(t,t_2)=0$ has at least two distinct solutions in the interval $ T(I_1)$.
If that is indeed the case, then by the mean-value theorem, there exists a point $t' \in T(I_1)$ such that $H'_t(t',t_2)=0$.
Since $H$ is continuously differentiable, the equality $H'_t(t',t_2)=0$ must hold for $\eps\to 0$, which implies that $t' \to t_1 - \overline{N}$ and $t_2 \to t_1$.
However,  $H'_t(t_1-\overline{N},t_1) = -G'(\overline{N})>0$ as already stated, and so we get a contradiction.

To conclude,  we proved that there exists an open interval $I_1= (t_1 - \eps, t_1 +\eps)$ such that for every $t_2 \in I_1$ there exists a unique solution $t\in T(I_1)$ for $J(t_1-t) = J^2(t_2-t)$,  and this solution is continuously differentiable in $t_2$.
We emphasize that both $J(\cdot)$ and $G(\cdot)$ are  independent of $t_1$ and $t_2$, and so is the \emph{length} of $I_1$ (and $I_0$).

We now revert to our original notations where $t_1$ and $t_2$ sustain any given capacity $p \in (0,1)$ in the one-stage and two-stage screening processes, respectively.
Since $t_2 < t_1$ and both values simultaneously converge to $\overline{V}+\overline{N}$ when $p \to 0$,  one can take $p_0>0$ such that $(t_2-\overline{N},\overline{V}) \subset I_1$, for every $p\in (0,p_0)$.
Moreover, the fact that both densities equal zero at the LHS of their supports (and $\underline{V_2} < \underline{V_1}$) suggests that $f_{V_2}(t) > f_{V_1}(t)$, in some $[t_2-\overline{N}, t_p)$, where $t_p$ is the single-crossing point for the chosen capacity.
This single-crossing point exists due to our construction, and due to the fact that both screening processes must obey the same capacity constraint (otherwise, $f_{V_2}(t) > f_{V_1}(t)$ for every $t$ in both supports). 
In other words, for every $t \in [t_2-\overline{N},\overline{V}]$,  it follows that $G(t_1-t)  \lesseqqgtr G^2(t_2-t)$ and $f_{V_1}(t) \lesseqqgtr f_{V_2}(t)$,  for $t \lesseqqgtr t_p$.
This establishes that $V_1$ (first-order) stochastically dominates $V_2$.
\hfill
\end{proof}

\subsection{Proof of Proposition \ref{Proposition - FOSD distinct screening}}
\begin{proof}
This proof follows a construction and arguments similar to those in the proof of Theorem \ref{Theorem - FOSD with fixed thresholds}.
Fix $\eps>0$,  random variables $(V,N)$,  and  some capacity $p \in (0,\eps)$.
Denote the one-stage and the two-stage screening problems by ${\rm{SP}}_i=(V,\{N\}_{k=1}^i,p)$,  where $i=1,2$,  respectively.
Denote the screening strategy (and threshold) of $\rm{SP}_1$ and $\rm{SP}_2$ by $\tau_1=(t_1)$ and  $\tau_2=(t_2^1,t_2^2)$, respectively.
With no loss of generality (as the stages are interchangeable), we assume that $t_2^1 > t_2^2$.
Similarly to the proof of Theorem \ref{Theorem - FOSD with fixed thresholds}, we have
\begin{eqnarray*}
f_{V_1}(t) & = & 
\begin{cases}
	\tfrac{1}{p} f_{V}(t) G\left( t_{1} - t \right),  & \text{ for } t_1 - \overline{N} \leq t \leq  \overline{V}, \\
	0, 							         										   & \text{ otherwise,} 
\end{cases}
\end{eqnarray*}
where $G(x) = \Pr(N \geq x)$ is a differentiable and decreasing function,  such that $G(x)=0$ if $x \geq \overline{N}$, and $G(x)=1$ if $x \leq \underline{N}$, and
\begin{eqnarray*}
f_{V_2}(t)
& = & 
\begin{cases}
	\tfrac{1}{p} f_{V}(t) G\left( t_{2}^1 - t \right)G\left( t_{2}^2 - t \right),  & \text{ for } t_2^1 - \overline{N} \leq t \leq  \overline{V}, \\
	0, 							         										   & \text{ otherwise.} 
\end{cases}
\end{eqnarray*}
The capacity constraint implies that $t_2^1 = \max\{t_2^1,t_2^1\} < t_1$,  which suggests that $\underline{V_2} < \underline{V_1}$ assuming that $\eps$ and $p$ are small such that $\underline{V_2} > \underline{V}$ (that is, assuming that the posteriors' lower bounds are strictly above the minimal level $\underline{V}$). 
Note that  $\underline{V_1} \to \overline{V}$ and $\underline{V_2} \to \overline{V}$ as $p \to 0$,  whereas $\overline{V_2} = \overline{V_1} =\overline{V}$ independently of $p$.
Namely, to reduce the capacity towards zero, the thresholds $ t_1$ and $t_2^1$ need to increase towards the upper bound of $\overline{V} + \overline{N}$.
Therefore,  as established in the proof of Theorem \ref{Theorem - FOSD with fixed thresholds}, the graphs $f_{V_1}(\cdot)$ and $f_{V_2}(\cdot)$ must intersect in at least one interior point, otherwise $f_{V_2}(t) > f_{V_1}(t)$ for every $t \in (t_2^1,\overline{V})$,  which violates the capacity constraint for one of the variables.

We will prove that $V_1$ {FOSD} $V_2$ by showing that there exists a range of values for $t_1$ close to $\overline{V} + \overline{N}$ (equivalently,  a range for $p$ close to zero), such that for every $t_2^1< t_1$ in that interval,  there exists exactly one (interior) point $t \in (t_1 - \overline{N},  \overline{V})$ so that $G(t_1-t) = G(t_2^1-t)G(t_2^2-t)$,  a single crossing between the two densities.

Consider the function $H(t,t_2^1|t_1,t_2^2)= G(t_1-t) - G(t_2^1-t)G(t_2^2-t)$. 
Again, with some abuse of notation, we allow for $t_1$ and $t_2^1$ to vary,  independently of $p$.
Assume, to the contrary, that for every interval $I=(\overline{V}+\overline{N}-c,\overline{V}+\overline{N})$ where $c>0$ and $t_1 \in I$, there exist  $t_2^1 \in (\overline{V}+\overline{N}-c,t_1)$  and a $t_2^2 < t_2^1$ which maintain the $\eps$-distinct property,  so that the equation $H(t,t_2^1|t_1,t_2^2)=0$ has at least two distinct solutions w.r.t.\ $t$ in the interval $I_1 = (t_1-\overline{N},\overline{V})$.
The fact that $t_2^2$ maintains the $\eps$-distinct property suggests that $p_2 = \Pr(X+N>t_2^2)< 1-\eps$,  where $X\sim V|\{V+N\geq t_2^1\}$.
Thus, there exists $\delta >0$ such that 
$$
t_2^2 > \underline{X}+\underline{N}+ \delta = t_2^1-\overline{N}+ \underline{N}+ \delta,
$$
since $\underline{X}= t_2^1-\overline{N}$, and 
\begin{equation} \label{Eq- lower bound on G}
G(t_2^2-t_1+\overline{N})< G(t_2^1-t_1+ \underline{N}+ \delta).
\end{equation}
By the mean-value theorem, there exists a point $t'\in I_1$ such that $H'_t(t',t_2^1|t_1,t_2^2)=0$.
Since $H$ is continuously differentiable, the equality $H'_t(t',t_2^1|t_1,t_2^2)=0$ must hold for $c\to 0$,  or equivalently, for $t' \to t_1 - \overline{N}$ and $t_2^1 \to t_1$.
However,  
\begin{eqnarray*}
H'_t(t_1-\overline{N},t_1|t_1,t_2^2) 
& = & -G'(\overline{N}) + G'(\overline{N})G(t_2^2-t_1+\overline{N}) + G(\overline{N})G'(t_2^2-t_1+\overline{N}) \\
& = & -G'(\overline{N}) + G'(\overline{N})G(t_2^2-t_1+\overline{N})  \\
& > & -G'(\overline{N}) + G'(\overline{N})G(t_2^1-t_1+ \underline{N}+ \delta) >0,
\end{eqnarray*}
where the first inequality follows from Eq.\ \eqref{Eq- lower bound on G} and the second inequality holds because $t_2^1$ is arbitrarily close to $t_1$ and $G(t_2^1-t_1+ \underline{N}+ \delta)<1$.

To conclude,  we proved that there exists an open interval $I=(\overline{V}+\overline{N}-c,\overline{V}+\overline{N})$,  where $c>0$ and $t_1\in I$, such that for every $t_2^1 \in (\overline{V}+\overline{N}-c,t_1)$ and $t_2^2< t_2^1$ which is $\eps$-distinct, there exists a unique solution $t \in (t_1 - \overline{N},  \overline{V})$ such that $G(t_1-t) = G(t_2^1-t)G(t_2^2-t)$.
In other words,  there exists $p_{\eps}>0$ so that for every $p \in (0,p_{\eps})$, the densities $f_{V_1}(t)$ and $f_{V_2}(t)$ coincide exactly once on their support (while $\underline{V_2} < \underline{V_1}$), thus establishing that $V_1$ (first-order) stochastically dominates $V_2$, as needed.
\hfill
\end{proof}

\subsection{Proof of Proposition \ref{Proposition - FOSD full distinct screening}}
\begin{proof}
This proof follows the pattern of  the proof of Proposition \ref{Proposition - FOSD distinct screening}, with the additional requirement to show that the lower bound $\underline{V_2}$ of the conditional impact variable under the two-stage screening is indeed lower than the one under the one-stage screening, i.e., $\underline{V_2} < \underline{V_1}$.

Fix $\eps \in (0,1)$ and two screening problems, ${\rm SP}_1=(V,\{N_1\},p)$ and ${\rm SP}_2=(V,\{N_1,N_2\},p)$, where $p<\eps$.
We start by proving that  $\underline{V_2} < \underline{V_1}$.
For this purpose consider the following claim

Note that $\tau=(t_2^1,t_2^2)$ is fully $\eps$-distinct, thus $p_2 \in (\eps,1-\eps)$ while $p_1 \cdot p_2 = p$.
Therefore, if $p$ is sufficiently small (close to zero; specifically, $p \ll \eps$),  then $p_1 \to 0$ and $t_2^1 \to \overline{V}+\overline{N}$.
So, from some point onward (namely, as long as $t_2^1$ is sufficiently close to $\overline{V}+\overline{N_1}$), the support of $V|\{V+N_1 \geq t_2^1\}$ shrinks towards $\overline{V}$,
and according to Lemma \ref{Lemma - short support, bounded threshold} below,  this suggests that $t_2^2$ is bounded away from $\underline{X} + \overline{N_1}$, where $x\sim V|\{V+N_1 \geq t_2^1\}$.
In other words, for every sufficiently small $p$, the lower bound $\underline{V_2}$ of $V_2$ is $t_2^1- \overline{N_1}$.
Recall that $t_2^1 < t_1$ (otherwise, $p_1<p$), so we conclude that for sufficiently small capacities $\underline{V_2} < \underline{V_1}$, as already stated.
So,  from this point onward, consider $p < \eps$ such that $|t_2^1 - (\overline{V}+\overline{N_1})|< \delta$, where $\Pr(N_1 \geq \overline{N_1}-\delta)< \eps$ for every fully $\eps$-distinct two-stage screening strategy.

A similar computation to the one presented in the proofs of Theorem \ref{Theorem - FOSD with fixed thresholds} and Proposition \ref{Proposition - FOSD distinct screening}  yields
\begin{eqnarray*}
f_{V_1}(t) & = & 
\begin{cases}
	\tfrac{1}{p} f_{V}(t) G_1\left( t_{1} - t \right),  & \text{ for } t_1 - \overline{N_1} \leq t \leq  \overline{V}, \\
	0, 							         										   & \text{ otherwise,} 
\end{cases}
\end{eqnarray*}
and
\begin{eqnarray*}
f_{V_2}(t)
& = & 
\begin{cases}
	\tfrac{1}{p} f_{V}(t) G_1\left( t_{2}^1 - t \right)G_2\left( t_{2}^2 - t \right),  & \text{ for } t_2^1 - \overline{N_1} \leq t \leq  \overline{V}, \\
	0, 							         										   & \text{ otherwise,} 
\end{cases}
\end{eqnarray*}
where $G_i(x) = \Pr(N_i \geq x)$ for $i=1,2$, with the same properties as discussed in the previous proofs.

The two screening processes can simultaneously support the capacity $p$ only if the two densities intersect at least once.
Let us now show that they do so exactly once.
Consider the function $H(t,t_2^1|t_1,t_2^2)= G_1(t_1-t) - G_1(t_2^1-t)G_2(t_2^2-t)$. 
With some abuse of notation, we allow for $t_1$ and $t_2^1$ to vary,  independently of $p$, while $t_2^2$ maintains the fully $\eps$-distinct property.

Assume, to the contrary, that for every interval $I=(\overline{V}+\overline{N_1}-c,\overline{V}+\overline{N_1})$, where $c>0$ and $t_1 \in I$, there exists a $t_2^1 \in (\overline{V}+\overline{N_1}-c,t_1)$ and a  $t_2^2 < t_2^1$ which maintains the fully $\eps$-distinct property,  such that the equation $H(t,t_2^1|t_1,t_2^2)=0$ has at least two distinct solutions w.r.t.\ $t$ in the interval $I_1 = (t_1-\overline{N_1},\overline{V})$.
By the mean-value theorem, there exists a point $t'\in I_1$ so that $H'_t(t',t_2^1|t_1,t_2^2)=0$.
Since $H$ is continuously differentiable, the equality $H'_t(t',t_2^1|t_1,t_2^2)=0$ must hold for $c\to 0$,  or equivalently, for $t' \to t_1 - \overline{N_1}$ and $t_2^1 \to t_1$.
However,  
\begin{eqnarray*}
H'_t(t_1-\overline{N},t_1|t_1,t_2^2) 
& = & -G_1'(\overline{N_1}) + G_1'(\overline{N_1})G_2(t_2^2-t_1+\overline{N_1}) + G_1(\overline{N_1})G_2'(t_2^2-t_1+\overline{N_1}) \\
& = & -G_1'(\overline{N_1}) + G_1'(\overline{N_1})G_2(t_2^2-t_1+\overline{N_1})  >0,
\end{eqnarray*}
where the inequality follows from fact that $\tau=(t_2^1,t_2^2)$ is a fully $\eps$-distinct strategy, and as such $G_2(t_2^2-t_1+\overline{N_1})$ is bounded away from $0$ and $1$.
Hence,  $H'_t(\cdot,t_2^1|t_1,t_2^2)>0$ remains strictly positive as $p$ tends to zero, and there is a single crossing between the two densities $f_{V_2}$ and $f_{V_1}$.
Since $\underline{V_2} < \underline{V_1}$, we conclude that $V_1$ FOSD $V_2$, as needed.
\hfill
\end{proof}

\begin{lemma} \label{Lemma - short support, bounded threshold}
For every noise variable $N$ and every $\eps>0$,  there exists $c>0$ such that for every impact variable $V$ supported on an interval $I$ of length $c'<c$,  the inequality $\Pr(V+N\geq t) >\eps$ implies that $t< \underline{V}+\overline{N}$.
\end{lemma}

\begin{proof}
Fix a noise variable $N$.
Take $N_0 \in {\rm Supp}(N)$ such that $\Pr(N \geq N_0) = \eps$.
Fix $c= \overline{N}-N_0$ and fix $a \in \mathbb{R}$.
Consider the interval $I=[a, a +\overline{N}- N_0]$ and an impact variable $V$ supported on $I$.
Take $t= \underline{V}+ \overline{N}= a+ \overline{N}$ and compute $\Pr(V+N \geq t)$ as follows:
\begin{eqnarray*} 
\Pr(V+N \geq a+\overline{N}) 
& = &  \int_{a}^{a+\overline{N}- N_0} f_{V}(x) \Pr\left(N \geq a+\overline{N} - x\right) dx \\
& < &  \int_{a}^{a+\overline{N}- N_0} f_{V}(x) \Pr\left(N \geq N_0 \right) dx \\
& = & \eps \int_{a}^{a+\overline{N}- N_0} f_{V}(x)  dx = \eps .
\end{eqnarray*}
Thus, for every $V$ supported on an interval whose length is $\overline{N}- N_0$,  the inequality $\Pr(V+N\geq t) >\eps$ implies that $t< \underline{V}+\overline{N}$.
Since the same computation holds for any interval of length smaller than $\overline{N}-N_0$, the statement holds.
\hfill
\end{proof}

\subsection{Proof of Theorem \ref{Theorem - two-stage dominance with fixed thresholds} }
\begin{proof}
Fix $V,N$, and a capacity $p > \Pr(V+N > \min \{  \overline{V} + \underline{N} ,  \underline{V} + \overline{N}\})$.
For every $i=1,2,$ denote the $i$-stage screening problem by ${\rm{SP}}_i=(V,\{N\}_{k=1}^i,p)$, and denote the screening strategy and threshold of $\rm{SP}_1$ and $\rm{SP}_2$ by $\tau_1=(t_1)$ and  $\tau_2=(t_2,t_2)$, respectively.
Note that for the given capacity (and for higher capacities as well),  $t_1 < \min \{  \overline{V} + \underline{N} ,  \underline{V} + \overline{N}\}$.
In addition, we know that $t_2 < t_1$ and both thresholds converge to $\underline{V}+\underline{N}$ as $p \to 1$.

Consider the construction presented in the proof of Theorem \ref{Theorem - FOSD with fixed thresholds}.
Under ${\rm SP}_1$ we get
\begin{eqnarray*}
f_{V_1}(t) & = & 
\begin{cases}
	\tfrac{1}{p} f_{V}(t) G\left( t_{1} - t \right),  & \text{ for }  \underline{V} \leq t \leq  t_1 - \underline{N}, \\
    \tfrac{1}{p} f_{V}(t),  & \text{ for }  t_1 - \underline{N} \leq t \leq  \overline{V}, \\
	0, 							         										   & \text{ otherwise,} 
\end{cases}
\end{eqnarray*}
where $G(x) = \Pr(N \geq x)$, and under ${\rm SP}_2$ we get
\begin{eqnarray*}
f_{V_2}(t) & = & 
\begin{cases}
	\tfrac{1}{p} f_{V}(t) G^2\left( t_{2} - t \right),  & \text{ for }  \underline{V} \leq t \leq  t_2 - \underline{N}, \\
    \tfrac{1}{p} f_{V}(t),  & \text{ for }  t_2 - \underline{N} \leq t \leq  \overline{V}, \\
	0, 							         										   & \text{ otherwise,} 
\end{cases}
\end{eqnarray*}
Since $t_2 < t_1$ along with the condition that densities are strictly positive, we can deduce that $f_{V_2}(t) \geq f_{V_1}(t)$ for every $t\in (t_2-\underline{N},\overline{V})$, and the inequality is strict for every $t\in (t_2-\underline{N}, t_1 - \underline{N})$.
Thus, we can establish the stochastic dominance of $V_2$ over $V_1$ by showing that there is a single crossing between the two densities in the interval $(\underline{V},t_2-\underline{N})$.

The fact that at least one crossing exists is trivial due to the capacity constraint.
Thus, let us assume, by contradiction, that there is more than one crossing independently of the high capacity $p$.
That is, we assume that for every $t_1 <\min \{  \overline{V} + \underline{N} ,  \underline{V} + \overline{N}\} $, there exists $t_2 \in (\underline{N}+\underline{V},t_1)$ such that the equation $H(t|t_1,t_2)=G(t_1-t)-G^2(t_2-t)=0$ has two solutions in the interval $(\underline{V},t_2-\underline{N})$.

Similarly to the proof of Theorem \ref{Theorem - FOSD with fixed thresholds},  the mean-value theorem ensure that there exists $t' \in (\underline{V},t_2-\underline{N})$ such that $H'(t|t_1,t_2)=-G'(t_1-t)+2G(t_2-t)G'(t_2-t)=0$.
By continuity,  this holds for every $t_1$ and $t_2$ close to $\underline{V}+ \underline{N}$,  so we can take the limit $t_i \to \underline{V}+\underline{N}$ and $t'\to \underline{V}$, which yields 
$$
0=-G'(\underline{N})+2G(\underline{N})G'(\underline{N})= -G'(\underline{N})+2G'(\underline{N})= G'(\underline{N}) < 0;
$$
we reached a contradiction.

Therefore, there exists $ p  > \Pr(V+N > \min \{  \overline{V} + \underline{N} ,  \underline{V} + \overline{N}\}) $ such that the densities  $f_{V_2}(t)$ and $f_{V_1}(t)$ intersect only once in the interval $(\underline{V},t_2-\underline{N})$, and this establishes that $V_2$ first-order stochastically dominates $V_1$, as needed.
\hfill
\end{proof}

\subsection{Proof of Proposition \ref{Proposition - increasing strategy}}
\begin{proof}
Fix an infinite screening problem ${\rm SP} = \left(V, N, p \right)$ with an increasing strategy $\tau = (t_1,t_2,\dots)$.
The strategy $\tau$ maintains, by definition, the stated capacity $p \in (0,1)$.
Therefore, $t_k \leq \overline{V}+\overline{N}$ for every $k$.
So, the sequence $(t_1,t_2,\dots)$ converges to some $t_{\infty} \in \R$.

Consider the stage-$k$ conditional distribution of $V$ in some $t\in \R$, denoted $F_{V_k}(t)$. 
Clearly,
\begin{eqnarray*}
F_{V_k}(t)
& = & \Pr \left(V_{k-1} \leq t | V_{k-1} + N_{k-1} \geq t_{k-1}\right) \\
& = & \tfrac{1}{p_k} \Pr \left(V_{k-1} \leq t , V_{k-1} + N_{k-1} \geq t_{k-1}\right) \\
& = & \tfrac{1}{p_k} \int_{-\infty}^t f_{V_{k-1}}(s) \Pr\left(N_{k-1} \geq t_{k-1} - s\right) ds \\
& = & \tfrac{1}{p_k} \int_{-\infty}^t f_{V_{k-1}}(s) \Pr\left(N \geq t_{k-1} - s\right) ds.
\end{eqnarray*}
Thus,
\begin{eqnarray*}
f_{V_k}(t) 
& = & \tfrac{1}{p_k}  f_{V_{k-1}}(t) \Pr\left(N \geq t_{k-1} - t\right) \\
& = & \tfrac{1}{p_{k-1}p_k}  f_{V_{k-2}}(t) \Pr\left(N \geq t_{k-2} - t\right) \Pr\left(N \geq t_{k-1} - t\right) \\
& = & \tfrac{1}{\prod_{i=1}^k p_i}  f_{V}(t) \prod_{i=1}^k \Pr\left(N \geq t_{i} - t\right),
\end{eqnarray*}
and 
$$f_{V_{{\rm SP}}(\tau)}(t) = \tfrac{1}{p}  f_{V}(t) \prod_{i=1}^{\infty} \Pr\left(N \geq t_{i} - t\right).$$

Fix $t\in \R$ and recall that $(t_i)_{i\in N}$ converges monotonically  to $t_{\infty}$.
If $t < t_{\infty} - \underline{N}$, then there exists $N_0$ such that $t < t_i - \underline{N}$ for every $i> N_0$, and  $\Pr\left(N \geq t_{i} - t \right) <1$.
Hence, $\prod_{i=1}^{\infty} \Pr\left(N \geq t_{i} - t\right)=0$.
Otherwise, $t \geq t_{\infty} - \underline{N}$, and $t \geq t_i - \underline{N}$ for every $i \in \N$.
This implies that $\Pr\left(N \geq t_{i} - t \right) =1$ for every $i \in \N$, and $\prod_{i=1}^{\infty} \Pr\left(N \geq t_{i} - t\right)=1$.
We thus conclude that
$$
f_{V_{{\rm SP}}(\tau)}(t)=
\begin{cases}
	0, 							& \text{ for every } t < t_{\infty} - \underline{N}, \\
	\tfrac{1}{p}f_{V}(t),  & \text{ for every } t \geq t_{\infty} - \underline{N}.
\end{cases}
$$
Clearly,
$$
\int_{\R}f_{V_{{\rm SP}}(\tau)}(t) dt = \int_{t_{\infty} - \underline{N}}^{\overline{V}}\tfrac{1}{p}f_{V}(t) dt =1,
$$
so $t_{\infty} = v_p + \underline{N}$, and $\tau$ induces a perfect screening.
\hfill
\end{proof}

\subsection{Proof of Theorem \ref{Theorem - stationary strategies converge to perfect screening}}
\begin{proof}
Fix an impact variable $V$, a noise variable $N$, and a capacity $p$.
For every $k\in \N$, let ${\rm SP}_k = (V,\{N\}_{i=1}^k,p)$ denote a $k$-stage screening problem, where $(V,N,p)$ are fixed, and consider the fixed-threshold strategy, $\tau = (t^k,\dots,t^k)$.
That is, $t^k$ denotes the fixed threshold of $\tau$ in ${\rm SP}_k$.

We begin by establishing that  $\{t^k\}_{k\geq 1}$ is a decreasing sequence that converges to  $v_p + \underline{N}$ as $k\to \infty$.
As shown in the proof of Theorem \ref{Proposition - increasing strategy}, for every $k \in \N$,
\begin{eqnarray*}
f_{V_k}(t) 
& = & \tfrac{1}{p}  f_{V}(t) \prod_{i=1}^k \Pr\left(N \geq t^k - t\right) = \tfrac{1}{p}  f_{V}(t) \left[ \Pr\left(N \geq t^k - t\right)\right]^k.
\end{eqnarray*}
So, if $t^k \leq t^{k+1}$, then 
\begin{eqnarray*}
f_{V_k}(t) 
& = & \tfrac{1}{p}  f_{V}(t) \left[ \Pr\left(N \geq t^k - t\right)\right]^k  \geq \tfrac{1}{p}  f_{V}(t) \left[ \Pr\left(N \geq t^{k+1} - t\right)\right]^{k+1} = f_{V_{k+1}}(t),
\end{eqnarray*}
and the inequality is strict for every $t$ where $f_{V}(t) \Pr\left(N \geq t^k - t\right)  >0$ and $ \Pr\left(N \geq t^k - t\right) < 1$.
Thus, we get a contradiction since $V_{k+1}$ is not normalized, and we deduce that $t^{k+1} < t^k$ for every $k\in \N$.
In addition, $t^k \geq \underline{V} + \underline{N}$ for every $k$; otherwise, there will be no screening whatsoever and the capacity constraint would be violated.
Hence, we conclude that $\{t^k\}_{k\geq 1}$ is a decreasing and bounded sequence, and so converges. 

To see that $\{t^k\}_{k\geq 1}$ converges to  $v_p + \underline{N}$, consider the infinite sequence $(v_p + \underline{N},v_p + \underline{N},\dots)$.
The induced conditional distribution of $V$ in the respective (infinite) screening problem ${\rm SP}$ is
\begin{eqnarray*}
f_{V_{\rm SP}}(t) 
& = & \tfrac{1}{p}  f_{V}(t) \prod_{i=1}^{\infty} \Pr\left(N \geq v_p + \underline{N} - t\right) =
\begin{cases}
	0, 	& \text{ for every } t < v_p, \\
	\tfrac{1}{p}f_{V}(t),  & \text{ for every } t \geq v_p.
\end{cases}
\end{eqnarray*}
If one would take an infinite strategy $(t,t,\dots)$ such that $t\neq v_p + \underline{N}$, the last computation shows that the induced capacity would not be $p$. 
So, we conclude that $\lim_{k\to \infty} t^k = v_p + \underline{N}$.

Finally, we need to prove that $V_k$ converges in distribution to $V|\{V\geq v_p\}$. 
Fix $t \in \R$.
If $t> v_p$, then there exists $k_t \in \N$ such that $t^k - \underline{N} < t$ for every $k\geq k_t$, and $\Pr\left(N \geq t^k - t\right)=1$.
Thus,  $f_{V_k}(t) =  \tfrac{1}{p}  f_{V}(t)$ for every $t> v_p$ and every $k\geq k_t$.
On the other hand, if $t < v_p$, then $\Pr(N \geq t^k - t) < \Pr (N \geq v_p + \underline{N} - t) <1$ for every $k\in \N$, which yields
$$
f_{V_k}(t) =  \tfrac{1}{p}  f_{V}(t) \left[ \Pr\left(N \geq t^k - t \right)\right]^k <  \tfrac{1}{p}  f_{V}(t) \left[  \Pr (N \geq v_p + \underline{N} - t)  \right]^k \to 0 \ \ \text{as} \ k \to \infty,
$$
as needed.

We now move on to the fixed-capacity strategy.
Again, for every $k\in \N$, let ${\rm SP}_k = (V,\{N\}_{i=1}^k,p)$ denote a $k$-stage screening problem where $(V,N,p)$ are fixed, and consider the fixed-capacity strategy, $\tau = (t_1^k,t_2^k,\dots, t_k^k)$.
One can easily show that, upon a screening stage, the induced posterior distribution of the impact variable first-order stochastically dominates the  prior distribution (in fact, the two distributions sustain the monotone likelihood-ratio property).
Thus, to maintain a capacity of $p^{1/k}$ in every stage, the thresholds need to strictly increase.

We start by proving that for every $\eps>0$ and every $l\in \N$ there exists $K_{l,\eps} \in \N$ such that for every $k> K_{l,\eps}$ there are at least $l^2$ stages where the threshold is strictly higher than $v_p + \underline{N} -\eps$, i.e., $| \{i : \ t_i^k > v_p+\underline{N} - \eps \} | \geq l^2$.

Fix $\eps>0$,  $l\in \N$, and  let $l_k = | \{i : \ t_i^k > v_p+\underline{N} - \eps \} |$ denote the number of stages in which the threshold is above  $v_p + \underline{N} -\eps$.
Note that the capacity up to stage $k-l_k$ is at least $\Pr(V > v_p - \eps)$ since values above $v_p-\eps$ pass the screening in these stages.
Denote $p_{\eps} = \Pr(V > v_p - \eps)$ and note that $1> p_{\eps} > p$.
In the remaining $l_k$ stages, the capacity is fixed to $p^{l_k / k}$ according to $\tau$.
Thus, when considering the overall capacity $p$, we get $p > p_{\eps} \cdot p^{l_k/k}$, which translates to $\tfrac{l_k}{k} > 1- \tfrac{\ln (p_{\eps}) }{ \ln (p)} >0$.
In other words, there exists $c>0$, which is independent of $k$, such that $l_k > c k$.
One can now fix $K_{l,\eps} = \tfrac{l^2}{c}$ to get the needed result.

Let us now use this claim to prove that $\tau$ converges to a perfect screening.
Fix $t< v_p$. 
Note $\delta =v_p - t$.
Fix $l>0$ so that $\Pr \left(N \geq \underline{N}+ \tfrac{\delta}{2} \right) < 1-\tfrac{1}{l}$. 
Take $K_{l,{\delta}/{2}}$ such the previous claim holds, and consider $k> K_{l, {\delta}/{2}}$.
Then, 
\begin{eqnarray*}
f_{V_k}(t) 
& = &  \tfrac{1}{p}  f_{V}(t) \prod_{i=1}^k \Pr\left(N \geq t^k - t \right) \\
& \leq &  \tfrac{1}{p}  f_{V}(t) \left[ \Pr\left(N \geq v_p + \underline{N} -\tfrac{\delta}{2} - t \right)\right]^{l^2} \\
%& = &  \tfrac{1}{p}  f_{V}(t) \left[ \Pr\left(N \geq v_p + \underline{N} -\tfrac{\delta}{2} - v_P + \delta \right)\right]^{l^2} \\
& = &  \tfrac{1}{p}  f_{V}(t) \left[ \Pr\left(N \geq  \underline{N} +\tfrac{\delta}{2}  \right)\right]^{l^2} \\
& \leq &  \tfrac{1}{p}  f_{V}(t) \left[  1 -\tfrac{1}{l} \right]^{l^2} \to 0
\end{eqnarray*}
as $l\to \infty$.
So, for every $t<v_p$, the density $f_{V_k}(t) $ converges to zero, which necessarily leads to a prefect screening.
\hfill
\end{proof}

\end{document}